\documentclass[final]{IEEEtran}
\IEEEoverridecommandlockouts
\usepackage{cite}
\usepackage{amsmath,amssymb,amsfonts,amsthm}
\usepackage{bbm}
\usepackage{algorithmic}
\usepackage{graphicx}
\usepackage{textcomp}
\usepackage{float}
\usepackage{stfloats}
\usepackage{xcolor}
\def\BibTeX{{\rm B\kern-.05em{\sc i\kern-.025em b}\kern-.08em
    T\kern-.1667em\lower.7ex\hbox{E}\kern-.125emX}}

\def\nb0{{\mathbf{0}}}
\def\nb1{{\mathbf{1}}}

\newtheorem{thm}{Theorem}

\def\figref#1{Fig.\,\ref{#1}}%
\def\eqnref#1{Eqn.\,\ref{#1}}%

\begin{document}

\title{Age of Positioning with Stochastic Motion Models}

\author{Wasif J. Hussain, Don-Roberts Emenonye, R. Michael Buehrer, Harpreet S. Dhillon \,\,\vspace{-2.1em}
		\thanks{The authors are with Wireless@VT, Virginia Tech, Blacksburg,
			VA 24061 USA (e-mail: {wasif,donroberts,rbuehrer,hdhillon}@vt.edu). The support of US NSF (grants CNS-1923807 and CNS-2107276) is gratefully acknowledged. }
	}

\maketitle

\begin{abstract}
Age of Information (AoI) is a key metric used for evaluating data freshness in communication networks, particularly in systems requiring real-time updates. In positioning applications, maintaining low AoI is critical for ensuring timely and accurate position estimation. This paper introduces an age-informed metric, which we term as Age of Positioning (AoP), that captures the temporal evolution of positioning accuracy for agents following random trajectories and sharing sporadic location updates. Using the widely adopted Random Waypoint (RWP) mobility model, which captures stochastic user movement through waypoint-based trajectories, we derive closed-form expressions for this metric under various queuing disciplines and different modes of operation of the agent. The analytical results are verified with numerical simulations, and the existence of optimal operating conditions is demonstrated.
\end{abstract}


\section{Introduction}


AoI has been studied as a key performance metric for quantifying information freshness, defined as the time elapsed since the most recent update was generated~\cite{kaul,pappas2023age}. AoI-based metrics have been applied across various wireless communication systems; see~\cite{pappas2023age} for a comprehensive overview. The most relevant examples to this paper include node selection in cognitive radar tracking~\cite{howard_aoi_tracking}, efficient Internet of Things (IoT) system design~\cite{md_aoi_mag}, resource allocation and trajectory optimization in UAV-aided communications~\cite{md_aoi_uav}, and optimizing trade-offs in joint sensing and communication systems~\cite{uav_aoi_poor}. In wireless positioning systems, accurately determining a target’s real-time location relies on the timely capture and processing of a sufficient number of measurements. However, position information can become outdated or entirely unavailable due to factors such as connectivity loss to servers, limited transmission resources, processing delays, and network cooperation constraints. An age-informed metric that can capture the temporal effects introduced by these factors on real-time positioning accuracy of an agent can help provide valuable insights for designing efficient positioning systems, which is the main focus of this paper.

Despite its importance, the role of AoI in wireless positioning systems has received little attention, apart from a couple of recent studies. In~\cite{wjh_mil}, the authors developed a Cram\'er-Rao Bound based framework to study the effect of age in the decay of information content of previously obtained range measurements for agents following deterministic motion models. In~\cite{enhancing_loc_awareness}, the authors introduced a new aging error of localization (AEOL) metric to incorporate the temporal decay of positioning accuracy which was quantified as the time averaged RMSE of position estimation for agents following rectilinear trajectory and a uniform distribution on RMSE. That said, a comprehensive analysis for non-deterministic motion models and various delay profiles remains an open area of research.

Inspired by this major gap, in this paper we propose a novel AoP metric in the context of agents following random trajectories. Our framework captures both errors due to estimation inaccuracy (spatial accuracy) and inaccuracy due to the fact that the mobile agent has moved since the last update (temporal decay). To model the agent's mobility we use the random waypoint (RWP) model, which is a commonly used stochastic mobility model for scenarios where node mobility patterns are unknown or highly dynamic, making it a popular choice for analyzing real-world performance for applications such as wireless sensor networks (WSNs), drone cellular networks and mobile ad hoc networks (MANETs), which require precise position estimation in mobile contexts~\cite{wsn_rwp,manet_rwp,path_plan_rwp_wsn,banagar2020performance}. 

\textit{Contributions}: We introduce a new mean squared error-based metric that captures the non-linear temporal decay in positioning accuracy caused by aged spatial measurements for an agent following the RWP mobility model. In contrast to the predominant linear age assumptions in the AoI literature, our formulation leads to a non-linear age model. This provides a concrete physical example where a non-linear age model is more appropriate~\cite{nl_aoi_pappas_icc,nl_aoi_pappas_trans}. Inspired by works on selective transmission under energy saving and limited scenarios\cite{arroyo2009optimal,arroyo2010optimal}, we introduce a probabilistic model according to which the agent shares sporadic location updates termed as \textit{polling}. These updates are processed sequentially and thus experience differential aging which has been modeled through a queue-based framework. We propose two different modes of operations of the agent based on our capability to keep track of its random movements. Using tools from queuing theory and stochastic processes, we derive closed form expression for AoP which is verified via numerical simulations. We demonstrate the existence of an optimal polling frequency for achieving optimal AoP. 

\section{System Model and Preliminaries}
\subsection{Random Waypoint Mobility Model}
\begin{figure}[!htbp]
    \centering\includegraphics[clip,trim=9cm 7cm 7cm 3cm,width=0.8\linewidth,keepaspectratio]{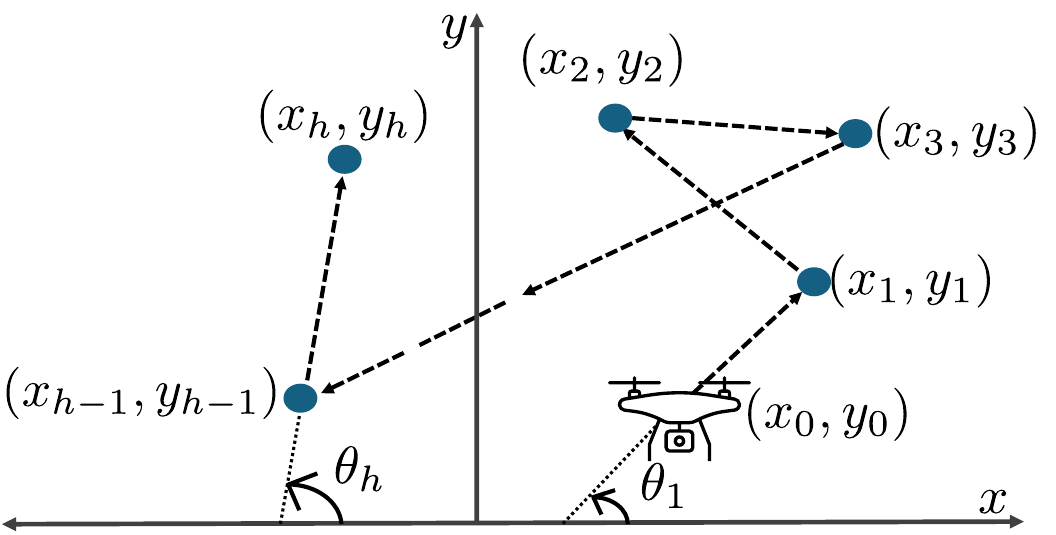}
    \caption{Illustration of the RWP model.} \label{fig:rwp_illus}
\end{figure}
As illustrated in~\figref{fig:rwp_illus}, we have a mobile agent operating in some area of interest whose position we want to track at a central server. The RWP model defines the agent’s movement in terms of waypoints, where it switches direction based on its latest measurement. According to this model, the agent starts at an initial location and determines its next waypoint by selecting a direction based on its most recent measurement, which we assume to be uniformly distributed in $(0,2\pi]$. To maintain an application-agnostic analysis, we assume a uniform distribution, which is sufficient for our purposes. However, our framework is not inherently restricted to this assumption, and the results can be extended to other distributions, {\em albeit} with potential loss of tractability. The agent then moves in this direction with a constant rectilinear velocity $v$, for a fixed time step $\delta$, completing what we refer to as a hop. At the end of each hop, the agent updates its direction with a new measurement, and this process repeats.


The agent's real-time position during the $h$-th hop is given as:
\begin{small}
\begin{equation}
    \mathbf{p_{n}}=
    \left[\begin{array}{cc}
         x_{n} \\
         y_{n}
    \end{array}\right]=\left[\begin{array}{cc}
         x_{o} \\
         y_{o}
    \end{array}\right]+\sum_{i=1}^{h} v\delta_i\left[\begin{array}{cc}
         \cos\theta_i \\
         \sin\theta_i
    \end{array}\right] ,
\end{equation}
\end{small}where, $\delta_i$ is the time spent in the $i$-th hop, $\mathbf{p_{\rm o}}=(x_o,y_o)$ and $\mathbf{\hat{p}_o}=(\hat{x}_o,\hat{y}_o)$ is the agent's true location and position estimate, respectively, when the last update was shared. Position estimation can be done by a number of techniques, such as time of arrival (ToA), received signal strength (RSS), etc\cite{handbook_of_loc}.

We propose two modes of operation of the agent:
\begin{itemize}
    \item \textbf{Dead Reckoning (DR) aided:} DR refers to estimating the position of an agent by extrapolating a past estimate based on inertial distance and direction feedback. In this mode, we assume the use of inertial direction sensors to track the agent's random direction switches $(\hat{\theta}_i)$, thus enabling DR. We assume $\hat{\theta_i}=\theta_i+\Delta$, where, $\Delta\sim \rm{Unif}\left(-\epsilon,\epsilon\right)$, and $\epsilon$ is small. 
    \item \textbf{Motion agnostic (MA) operation:} In this mode, we are completely unaware of the agent's random movement in between consecutive hops. We rely on the last obtained position estimate ($\mathbf{\hat{p}_o}$) as the agent's real-time position estimate ($\mathbf{\hat{p}_n}$). Encompassing a higher level of uncertainty and only minor analytical differences, this mode will be the focus of our discussion for the rest of the paper.
\end{itemize}For a unified analysis, all the results will be derived in terms of a parameter $c$ , where $c = 1$ corresponds to DR and $c = 0$ to MA. In terms of $c$ , the agent's real-time position estimate is given as: \begin{small}
    \begin{equation}
    \mathbf{\hat{p}_n}=\mathbf{\hat{p}_o}+ c\sum_{i=1}^h v\delta_i \left[\begin{array}{cc}
         \cos\hat{\theta}_i \\
         \sin\hat{\theta}_i
    \end{array}\right].
    \end{equation}
    \end{small}The mean squared error between the estimated position and the true position is given as: 
    \begin{small}
    \begin{equation}
        \label{eqn:aop_formalism}
        \begin{split}
        {\bf MSE}=&\mathbb{E}\left[(\mathbf{p_n}-\mathbf{\hat{p}_n})(\mathbf{p_n}-\mathbf{\hat{p}_n})^{\rm T}\right] \\
        =& \underbrace{\mathbb{E} \left[\begin{array}{cc}
             \left(x_o-\hat{x}_o\right)^2 & \left(x_o-\hat{x}_o\right)\left(y_o -\hat{y}_o\right)  \\
             \left(x_o-\hat{x}_o\right)\left(y_o -\hat{y}_o\right) & \left(y_o -\hat{y}_o\right)^2 
        \end{array}\right]}_{\mathbf{C_{\rm o}}} \\
        & + \sum_{i=1}^{h} v^2\delta_i^2 \mathbb{E}\left[\begin{array}{cc}
      c_1^2& c_1c_2  \\
      c_1c_2& c_2^2
         \end{array}\right],
         \end{split}
    \end{equation}
    \end{small}where, $c_1=\left(\cos\theta_i - c\cos\hat{\theta}_i\right)$, $c_2=\left(\sin\theta_i-c\sin\hat{\theta}_i\right)$, and $\mathbf{C}_{\rm o}$ denotes the covariance matrix for any unbiased position estimator. 
    We use the trace of mean squared error matrix, which we term as position error bound (PEB), as a metric to quantify position uncertainty given as:
\begin{small}
\begin{equation}
\begin{split}
\label{eq:peb_expr}
    {\rm PEB}&={\rm Trace}(\mathbf{MSE})
    = {\rm PEB}_{\rm o} + \left(1+c^2-2c\mathbb{E}\left[\cos\Delta\right]
    \right)\sum_{i=1}^h v^2\delta_i^2 \\
    &={\rm PEB}_{\rm o}+\left(1-c+c\mathbb{E}\left[\Delta^2\right]\right)\sum_{i=1}^h v^2\delta_i^2={\rm PEB}_{\rm o}+\kappa\sum_{i=1}^h v^2\delta_i^2,
\end{split}
\end{equation}
\end{small}where, $\kappa=1$ for $c=0$, and $\kappa=\frac{\epsilon^2}{3}$ for $c=1$, respectively.

\subsection{Queuing Model}
In this work, we analyze FCFS (first come first served) queue disciplines to model the generation and availability instants for updates associated with the mobile agent. Each update packet enters a queue waiting for resources to be transmitted to the central server. For operational efficiency, a packet is chosen for transmission with a probability $p$, otherwise discarded. The packet then travels through a wireless channel before it becomes available at the central server.
The disciplines we analyze are namely M/M/1, D/M/1, M/D/1 and D/D/1 respectively~\cite{nelson2013probability}. Where ``D" and ``M" denote a deterministic and Markovian process respectively. We use $\lambda$ and $\mu$ to denote the rate for arrival and departure processes respectively. For deterministic queues, define, $ D_s=\frac{1}{\mu}$ and $ D_a=\frac{1}{\lambda}$ as the mean update arrival time and mean service time respectively. The mean utilization for the server processing these updates is denoted by $ \rho=\frac{p\lambda}{\mu}$.  

\section{Age of Positioning}
\subsection{A new metric}
\begin{figure}[!htbp]
	 \centering\includegraphics[clip,trim=2.48cm 0.38cm 2cm 3.05cm, width=0.85\linewidth, keepaspectratio]{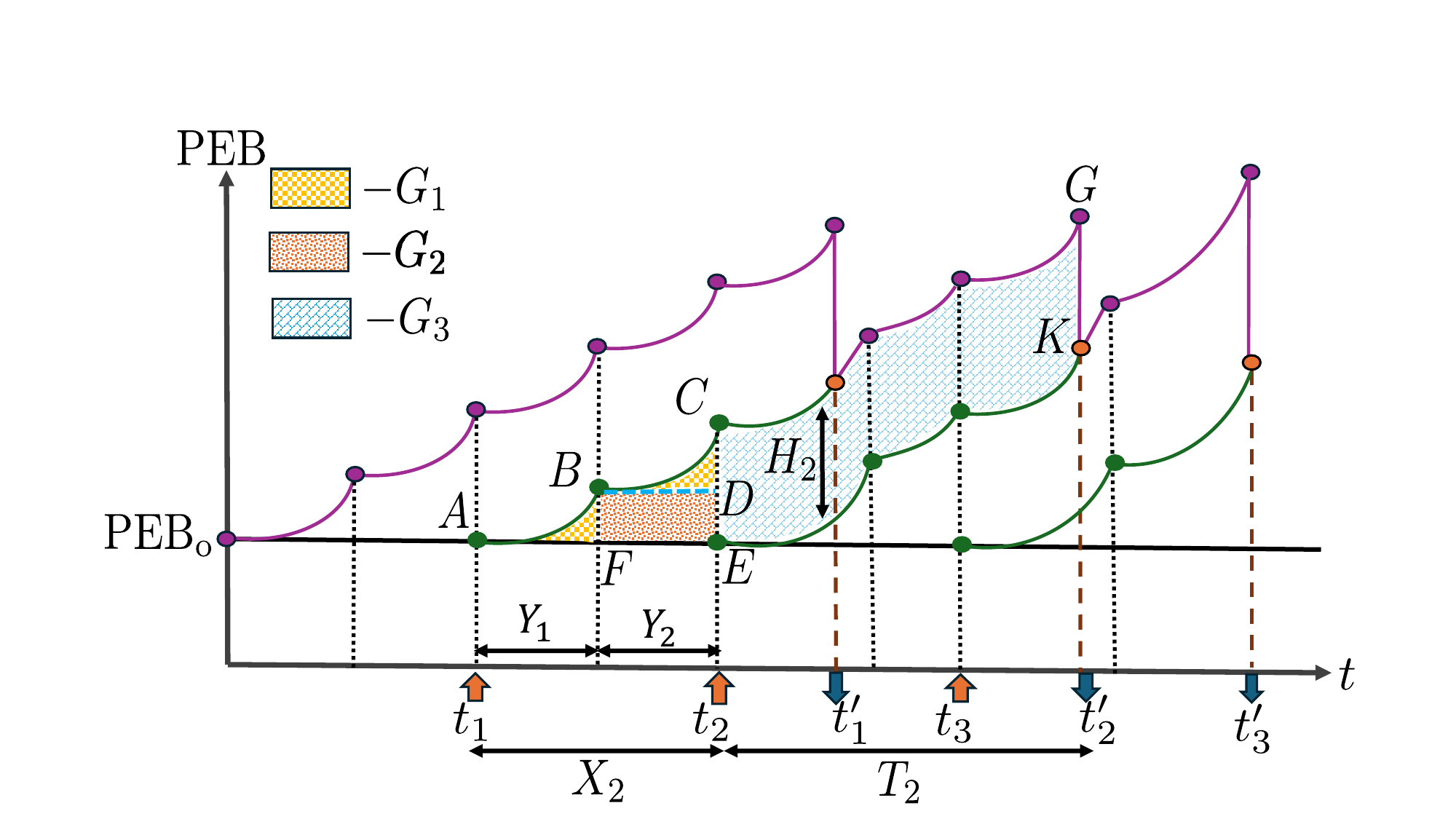}
		 \caption{Variation of PEB with age. The dotted black lines denotes new hop instants and the dotted brown lines denote the availability instants at the central server.} \label{fig:aop_main}
\end{figure}

In this section, we introduce the proposed AoP metric, which quantifies a mean positioning error for a semi-cooperative agent following RWP trajectory. 
Updates generated at waypoints are modeled as an arrival process, and the instant they become available at the central server constitute the departure process. In~\figref{fig:aop_main}, we plot the evolution of PEB with age. The upward orange arrows mark the instants when position updates were obtained. The downward blue arrows mark the instant the information became available at the central server. As we can see from this figure, by the time these updates become available, the agent has moved to a completely different position. As evident from~\eqnref{eq:peb_expr}, the PEB grows parabolically in between consecutive hops, indicating non-linear aging.

We define AoP as the time-averaged difference between PEB using an aged estimate (or DR estimate) and assuming instant estimation, over a certain observation interval $\mathcal{T}$, given as:
\begin{small}
\begin{equation}
\label{eq:aop_int}
    {\rm AoP}=\frac{1}{\mathcal{T}} \int_0^{\mathcal{T}} {\rm } \left({\rm PEB}(t) - {\rm PEB}_{\rm o}\right) {\rm d}t.
\end{equation}
\end{small}

For MA operation, this metric captures how well the last processed position estimate remains valid until a new update is obtained and processed. For the other operation, it captures the accuracy of DR estimation. To evaluate the integral in~\eqnref{eq:aop_int}, we decompose the enclosed area into disjoint areas $Q_i \; \left(A\rightarrow B \rightarrow C \rightarrow G \rightarrow K \rightarrow E \rightarrow F \rightarrow A\right)$. Further we divide each of these areas into sub-areas $G_1,G_2,G_3$ as shown in~\figref{fig:aop_main}. Let $X_i$ denote the time interval between the $(i-1)$-th and $i$-th update, and $T_i$ as the sum of queuing and waiting times that the update spends before the position information becomes available at the central server. $\left\{Y_j\right\}_{j=1}^k$ denotes the time interval between successive hops in which updates are not polled, {\em i.e,} $X_i=\sum_{j=1}^k Y_j$.

The average age in~\eqnref{eq:aop_int} can be expressed as follows~\cite{kaul}:
\begin{small}
\begin{equation}
    {\rm AoP}=p \kappa \lambda \mathbb{E}\left[Q_i\right].
\end{equation}
\end{small}
\subsection{AoP with queuing}
In this section, we derive and analyze the AoP for different queuing scenarios.
The area under each segment $\left\{G_i\right\}$ can be written as:
\begin{small}
\begin{equation}
    \begin{split}                   
        &\mathbb{E}\left[G_1\right]=\mathbb{E}\left[\sum_{j=1}^k \frac{v^2}{3} Y_j^3\right], \;\; \mathbb{E}\left[G_3\right] = \mathbb{E} \left[H_i T_i\right], \\
        &\mathbb{E}\left[G_2\right]=\mathbb{E} \left[\sum_{j=1}^k Y_j.\left(\sum_{l=1}^{j-1}v^2 Y_l^2\right)\right],
    \end{split}
\end{equation}
\end{small}where, $H_i=\sum_{j=1}^k v^2 Y_j^2$.

\begin{thm}
    For $\rho<1$, the AoP for an agent following RWP trajectory under M/M/1 queuing can be computed using:
    \begin{small}
    \begin{equation}
    \begin{split}
        \mathbb{E}\left[Q_i\right]&= 2 v^2 \frac{(1-p)}{p^2 \lambda^3} + \frac{2v^2}{p\lambda^3} + \frac{2v^2}{p\lambda^2\mu} \\
        &+ \frac{2v^2 p \rho_e}{\mu(1-\rho)}\frac{1}{(\mu(1-\rho)+\lambda)^2 (1-\rho_e(1-p))^2},
        \end{split}
    \end{equation}
    \end{small}where, $ \rho_e=\frac{\lambda}{\mu(1-\rho)+\lambda}$.
\end{thm}
\begin{proof}
    Refer to Appendix-A.
\end{proof}


\begin{thm}
    For $D_s<D_a$, the AoP for an agent following RWP trajectory under D/D/1 queuing can be computed using:
    \begin{small}
    \begin{equation}
        \mathbb{E}\left[Q_i\right]= \frac{v^2 D_a^2 D_s}{p} +\frac{v^2 D_a^3}{3p} +\frac{v^2 D_a^3 (1-p)}{p^2}. 
    \end{equation}
    \end{small}
\end{thm}
\begin{proof}
    Refer to Appendix-B.
\end{proof}


\begin{thm}
    The AoP for an agent following RWP trajectory under D/M/1 queuing can be computed using:
    \begin{small}
    \begin{equation}
        \begin{split}
            \mathbbm{E}[Q_i] &= \frac{v^2 D_a^3}{3p} + \frac{v^2 D_a^3 (1-p)}{p^2} \\
            & + v^2 D_a^2\left[\frac{p e^{-\mu(1-\beta)D_a}}{\mu(1-\beta)\left\{1-(1-p)e^{-\mu(1-\beta)D_a}\right\}^2}+\frac{1}{\mu p}\right],
        \end{split}
    \end{equation}
    \end{small}where, $\beta$ is the solution of the following equation:
    \begin{small}
\begin{equation}
    \beta=\frac{p e^{-\mu(1-\beta)D_a}}{1-(1-p)e^{-\mu(1-\beta)D_a}}.
\end{equation}
\end{small}
\end{thm}
\begin{proof}
    Refer to Appendix-C.
\end{proof}
\begin{figure*}[!b]
\noindent\rule{\textwidth}{0.5pt}
\begin{small}
    \begin{equation}
    \label{eq:Tk}
        \begin{split}
            &T(k)=\int \dots \int \lambda^k g(y_1,y_2,\dots,y_k)(v^2 \sum_{j=1}^k y_j^2) e^{-\lambda\sum_{j=1}^k y_j} {\rm d}y_1 {\rm d}y_2\dots {\rm d}y_k, \\
            & g(y_1,\dots,y_k) = \frac{D_s \rho}{2(1-\rho)} +D_s -\sum_{j=1}^k y_j -\mathbbm{1}\left(\sum_{j=1}^k y_j \geq D_s \right) \int_0^{\sum_{j=1}^k y_j -D_s} (w+D_s-\sum_{j=1}^k y_j)f_W(w){\rm d}w ,\\
            & F_W(w)=(1-\rho) \sum_{k=0}^m e^{p\lambda(w-kD_s)}\frac{\left\{-p\lambda(w-kD_s)\right\}^{k}}{k!}; m=\left\lfloor\frac{w}{D_s}\right\rfloor.
        \end{split}
    \end{equation}
    \end{small}
\end{figure*}
\begin{thm}
    For $\rho<1$, the AoP for an agent following RWP trajectory under M/D/1 queuing can be computed using:
    \begin{small}
    \begin{equation}
        \mathbb{E}\left[Q_i\right] =\frac{2v^2}{p\lambda^3} + \frac{2v^2(1-p)}{p^2\lambda^3} + \frac{2v^2 D_s}{p\lambda^2}  + \sum_{k=1}^{\infty} p(1-p)^{k-1} T(k),
    \end{equation}
    \end{small}where, $T(k)$ is defined in~\eqnref{eq:Tk}. The probability distribution function (PDF) of waiting time in M/D/1 queue is obtained numerically by taking the derivative of the cumulative density function (CDF)\cite{franx_md1}.
\end{thm}
\begin{proof}
    Refer to Appendix-D.
\end{proof}
\section{Numerical Results}
In this section, we perform simulations to verify our analyses and discuss key insights on system-level performance. For baseline PEB, the choice of $\mathbf{C}_{\rm o}$
does not impact the numerical simulations since the dependent term is eliminated when taking the difference in~\eqnref{eq:peb_expr}. The velocity of agent was taken as $5 {\rm ms}^{-1}$. For server utilization, we fixed $\mu=20$, $\lambda=20$, and varied the polling probability $p$. This results in an average service time of roughly $50 {\rm ms}$ and an average update interval of roughly $5-500 {\rm ms}$, which covers a wide range of practical localization scenarios~\cite{enhancing_loc_awareness}.

\figref{fig:aop_combined} presents the analytical AoP results against different server utilizations for the MA mode of operation. The result for DR-aided operation is a scaled (by $\kappa$) version of this and is not provided due to space constraints. As expected, the Monte Carlo simulations closely follow our analytical results. We also observe that at lower server utilizations, the AoP is primarily influenced by the arrival process. This happens because fewer packets arrive and thus most of the packets are immediately serviced, leading to negligible waiting times. As the server utilization is increased, there are more packets waiting for service, thus the service process influences the age more. We can observe Markovian service processes to observe a higher AoP when the utilization is high. This is because in Markovian queues there are some packets which experience longer than average service times which results in a higher overall AoP.

In general, we expect AoP to be higher when the arrival frequency ($p\lambda$) is low, due to fewer updates. Additionally, as the arrival rate increases, packets experience longer waiting times for service, leading to increased aging and consequentially a higher AoP. For M/M/1, D/M/1 and M/D/1 queues we can observe the existence of an optimum operating point. In case of D/D/1 queue, the optimum occurs when polling probability tends to $1$, this is because updates experience no waiting time when $\lambda<\mu$. 

\begin{figure}[!htbp]
    \centering               \includegraphics[width=0.85\linewidth,keepaspectratio]{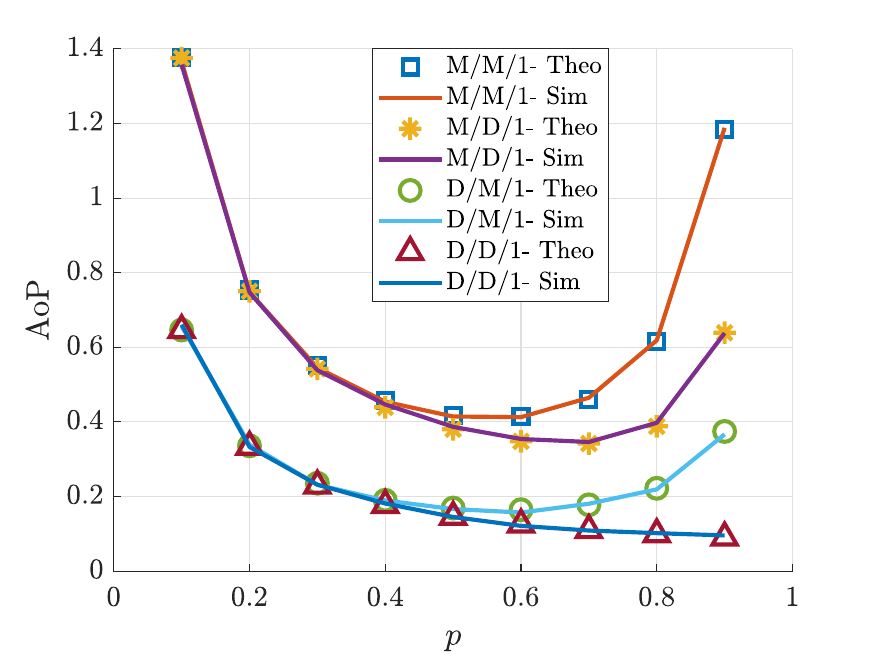}
    \caption{Verification of analytical results for M/M/1, D/M/1, M/D/1 and D/D/1 disciplines for MA operation of the agent. } \label{fig:aop_combined}
\end{figure}

\section{Conclusions}

In this work, we introduced a novel AoP metric which characterizes the positioning performance for a semi-cooperative agent following RWP trajectory with different levels of knowledge of the agent's random movements. Our metric jointly captures the impact of both spatial estimation error and temporal decay due to the agent's mobility on its positioning accuracy. This is the first work that has explored the interplay of age with positioning precision for stochastic trajectories. We used the RWP model, which is a popular choice for agent motion modeling making it relevant to a broad range of application scenarios such as aerial communications, WSNs, and MANETs. We developed a queue-based framework to account for factors that delay position-related information. We derived a closed-form expressions to characterize AoP under different FCFS queuing models and verified them with numerical simulations. We demonstrated the existence of an optimum polling probability for achieving the best AoP. 

This work has potential applications in numerous areas, such as age-aware scheduling and resource allocation for aerial communication systems, as well as improving energy efficiency in large-scale WSNs. Another promising direction for exploration is the analysis of preemptive queuing disciplines, where only the latest update packet is retained while all others are discarded.

\section*{Appendix-A}
We begin by evaluating the mean of each of the sub-areas individually. We start with $G_1$:
\begin{small}
\begin{equation}
    \begin{split}
        \mathbb{E}[G_1]&=\frac{v^2}{3}\mathbb{E}\left[\sum_{j=1}^k Y_j^3\right]=\frac{v^2}{3} \mathbb{E}_k\left[\sum_{j=1}^k\mathbb{E}_{Y}\left[Y_j^3\right]\right] \\
        &\overset{(a)}{=}\frac{v^2}{3} \sum_{k=1}^\infty p(1-p)^{k-1} \left[k \frac{6}{\lambda^3}\right] 
        \overset{(b)}{=} \frac{2v^2}{\lambda^3 p},
    \end{split}
\end{equation}
\end{small}where, $(a)$ follows from intervals being independent, and exponentially distributed in a Poisson process, and the definition of gamma integral, and $(b)$ uses the partial sum formula.
Next we evaluate the mean area covered under $G_2$ as follows:
\begin{small}
\begin{equation*}
    \begin{split}
    \mathbb{E}\left[G_2\right]&=\mathbb{E} \left[\sum_{j=1}^k Y_j.\left(\sum_{l=1}^{j-1}v^2 Y_l^2\right)\right] \\
    &=v^2 \mathbb{E}_k\left[\sum_{j=1}^k \sum_{l=1}^{j-1} \mathbb{E}\left[Y_jY_l^2\right]\right] = v^2 \mathbb{E}_k\left[\frac{k(k-1)}{2}\frac{1}{\lambda} \frac{2}{\lambda^2}\right] \\     
    & = \frac{v^2 p}{(1-p)\lambda^3} \sum_{k=1}^\infty k(k-1)(1-p)^k = \frac{2v^2(1-p)}{p^2\lambda^3}.
    \end{split}
\end{equation*}
\end{small}Finally, we compute the mean area covered under $G_3$ as follows:
\begin{small}
\begin{equation*}
    \begin{split}
        \mathbb{E}\left[G_3\right]&= \mathbb{E}[H_iT_i] = \mathbb{E}[H_i(S_i+W_i)] \\
\end{split}
\end{equation*}
\begin{equation}
\begin{split}
        &= v^2 \underbrace{\mathbb{E}\left[(T_{i-1}-X_i)^+\sum_{j=1}^k Y_j^2 \right]}_{\rm Term1} + v^2 \underbrace{\mathbb{E}\left[S_i \sum_{j=1}^kY_j^2\right]}_{\rm Term2}.
    \end{split}
\end{equation}
\end{small}The individual terms are given as:
\begin{small}
\begin{equation*}
\begin{split}
        {\rm Term2}&=\frac{1}{\mu}\mathbb{E}\left[\sum_{j=1}^k Y_j^2\right]= \frac{2v^2}{\mu\lambda^2 p},\\
        {\rm Term1}&=\mathbb{E}_{k,\{Y_j\}}\left[\left(T_{i-1}-\sum_{j=1}^k Y_j\right)^+ \sum_{j=1}^k Y_j^2\right] \\
        &\overset{(c)}{=}\mathbb{E}_k\left[ \mathbb{E}_{\{Y_j\}_{j=1}^k} \left[\sum_{j=1}^k Y_j^2 \mathbb{E}_{T|\{Y_j\}_{j=1}^k}\left[\left(T-\sum_{j=1}^k Y_j\right)^+\right]\right] \right] \\
        & \overset{(d)}{=} \mathbb{E}_k\left[ \mathbb{E}_{\{Y_j\}_{j=1}^k} \left[\left(\sum_{j=1}^k Y_j^2\right) \frac{e^{-\mu(1-\rho)\sum_{j=1}^k Y_j}}{\mu(1-\rho)} \right] \right] \\
        & \left.\frac{e^{-\left[\mu(1-\rho)+\lambda\right]\sum_{j=1}^k y_j}}{\mu(1-\rho)}  \left[\mu(1-\rho)+\lambda\right]^k {\rm d}y_1 \dots {\rm d}y_k\right] \\
        & = \mathbb{E}_k\left[k \left[\frac{\lambda}{\mu(1-\rho)+\lambda}\right]^k \frac{1}{\mu(1-\rho)} \right.\cdot \\ &\left .\int_0^\infty y^2 \left[\mu(1-\rho)+\lambda\right]e^{-\left(\mu(1-\rho)+\lambda\right)} {\rm d}y\right]  
        \end{split}
\end{equation*}
\begin{equation*}
    \begin{split}
        &\overset{(e)}{=}\frac{2p}{\mu(1-\rho)(1-p)\left[\mu(1-\rho)+\lambda\right]^2}\left(\sum_{k=1}^\infty k \rho_e^k(1-p)^k\right) \\
        &=\frac{2\rho_e p}{\mu(1-\rho)\left[\mu(1-\rho)+\lambda\right]^2}\frac{1}{\left[1-\rho_e(1-p)\right]^2},
    \end{split}
\end{equation*}
\end{small}where, $(c)$ follows from nested expectation theorem and by writing the waiting time as $W_i=(T_{i-1}-X_i)^+$~\cite{kaul}. For $(d)$, we use the pdf of response time for a M/M/1 queue when it reaches stationary state which is an exponential distribution with rate $\mu(1-\rho)$~\cite{stewart_book}. For $(e)$, we defined $ \rho_e=\frac{\lambda}{\mu(1-\rho)+\lambda}$.
\section*{Appendix-B}
We evaluate the different sub-areas as follows:
\begin{small}
\begin{equation*}
\begin{split}
        \mathbb{E}[G_1]&=\frac{v^2}{3}\mathbb{E}\left[\sum_{j=1}^k Y_j^3\right]=\frac{v^2}{3} \mathbb{E}_k\left[k D_a^3\right]=\frac{v^2 D_a^3}{3p},  \\
        \mathbb{E}\left[G_2\right]&=\mathbb{E} \left[\sum_{j=1}^k Y_j.\left(\sum_{l=1}^{j-1}v^2 Y_l^2\right)\right] \\
        &=\frac{v^2 D_a^3 p}{2(1-p)} \sum_{k=1}^\infty \left[k(k-1)(1-p)^k\right] =\frac{v^2 D_a^3 (1-p)}{p^2} ,\\
        \mathbb{E}[G_3]&=\mathbb{E}[H_iT_i]= \mathbb{E}\left[(\sum_{j=1}^k v^2 D_a^2)(W_i+S_i)\right] \\
        \end{split}
        \end{equation*}
\begin{equation*}
    \begin{split}
        &\overset{(a)}{=} v^2 D_a^2 D_s \sum_{k=1}^\infty k p (1-p)^{k-1} = \frac{v^2 D_a^2 D_s}{p},
    \end{split}
\end{equation*}
\end{small}where, $(a)$ follows from the fact that $W_i=0$ for $D_s\leq D_a$. 

\section*{Appendix-C}
We can compute $\displaystyle\mathbb{E}[G_1]$ and $\displaystyle \mathbb{E}[G_2]$ in a similar way to D/D/1 queue. For computing $\displaystyle\mathbb{E}[G_3]$, we proceed as follows:
\begin{small}
\begin{equation}
    \begin{split}
        \mathbb{E}[G_3] &=\mathbb{E}[H_iT_i]= \mathbb{E}\left[(\sum_{j=1}^k v^2 Y_j^2)(W_i+S_i)\right] \\
        &=v^2 D_a^2 \underbrace{\mathbb{E}\left[kW_i\right]}_{\rm Term1} + 
        v^2 D_a^2 \underbrace{\mathbb{E}\left[kS_i\right]}_{\rm Term2}.
    \end{split}
\end{equation}
\end{small}The individual terms are given as:
\begin{small}
\begin{equation*}
    \begin{split}
        {\rm Term1} &= \mathbb{E}_k \left[k \mathbb{E}\left[W_i|K=k\right]\right] \\
        &=\mathbb{E}_k\left[k\mathbb{E}\left[(T_{i-1}-kD_s)^+|K=k\right]\right]\overset{(a)}{=} \mathbb{E}_k\left[\frac{e^{-\mu(1-\beta)kD_s}}{\mu(1-\beta)}\right]\\
        &=\frac{p}{\mu(1-\beta)(1-p)} \frac{(1-p)e^{-\mu(1-\beta)D_s}}{\left(1-(1-p)e^{-\mu(1-\beta)D_s}\right)^2}, \\
    {\rm Term2}&=\frac{v^2D_a^2}{p\mu},
\end{split}
\end{equation*}
\end{small}where, $(a)$ follows from the pdf of system time for a A/M/1 queue, A can be any arrival process~\cite{nelson2013probability}. For obtaining $\beta$, we begin by writing the distribution of inter-arrival times as $f_A(x)= p(1-p)^{i-1}\delta(x-iD_a) \;\; \forall i\in \mathcal{N}$. We can thus solve for $\beta$ as the solution to the following equation~\cite{nelson2013probability}:
\begin{small}
\begin{equation}
   \beta=\frac{p e^{-\mu(1-\beta)D_a}}{1-(1-p)e^{-\mu(1-\beta)D_a}}.
\end{equation}
\end{small}
\section*{Appendix-D}
We can compute $\mathbb{E}[G_1]$ and $\mathbb{E}[G_2]$ in a similar way to M/M/1 queue. For computing $\mathbb{E}[G_3]$ we proceed as follows:
\begin{small}
\begin{equation}
    \begin{split}
        &\mathbb{E}[G_3]=\mathbb{E}[H_i T_i] =v^2\underbrace{\mathbb{E}\left[\left(\sum_{j=1}^k Y_j^2\right)W_i\right]}_{\rm Term1} + 
        v^2\underbrace{D_s\mathbb{E}\left[\sum_{j=1}^k Y_j^2\right]}_{\rm Term2}.
    \end{split}
\end{equation}
\end{small}The individual terms are given as:
\begin{small}
\begin{equation*}
    \begin{split}
        {\rm Term2}&=\frac{2 v^2 D_s}{p\lambda^2}, \\
        \end{split}
        \end{equation*}
        \begin{equation*}
        \begin{split}
        {\rm Term1}&=\mathbb{E}\left[\left(\sum_{j=1}^k Y_j^2\right)\left(T_{i-1}-\sum_{j=1}^k Y_j\right)^+\right] \\
        & =\mathbb{E}_{\{Y_j\}}\left[\left(\sum_{j=1}^k Y_j^2\right) \mathbb{E}_{T|\{Y_j\}}\left[\left(T-\sum_{j=1}^k Y_j\right)^+\right]\right] \\
     & =\mathbb{E}_{\{Y_j\}}\left[\left(\sum_{j=1}^k Y_j^2\right) \underbrace{\mathbb{E}_W\left[\left(W+D_s-\sum_{j=1}^k Y_j\right)^+\right]}_{\rm g(\{Y_j\}_{j=1}^K)}\right]. \\
     \end{split}
\end{equation*}
\end{small}

\bibliographystyle{IEEEtran}
\bibliography{bib}
\vspace{12pt}

\end{document}